\documentclass[journal]{IEEEtran}
\usepackage[pdftex,breaklinks,colorlinks,
  citecolor=blue,
  urlcolor=blue]{hyperref}

\usepackage{cite}
\ifCLASSINFOpdf
   \usepackage[pdftex]{graphicx}
\else
  \usepackage{graphicx}
\fi
\graphicspath{{figures/}}

\usepackage{amsmath}
\usepackage{empheq}
\usepackage{bbm}
\usepackage{amssymb}
\usepackage{mathrsfs}  
\usepackage{amsthm}
\usepackage{bm}

\usepackage{array}

\usepackage{url}

\theoremstyle{plain}
\newtheorem{theorem}{Theorem}[section]
\newtheorem{proposition}[theorem]{Proposition}
\newtheorem{lemma}[theorem]{Lemma}
\newtheorem{corollary}[theorem]{Corollary}

\theoremstyle{definition}

\theoremstyle{remark}
\newtheorem{remark}[theorem]{Remark}

\newcommand{\E}{\mathbb{E}}
\newcommand{\R}{\mathbb{R}}

\newcommand{\Z}{\mathbb{Z}}
\newcommand{\N}{\mathbb{N}}
\newcommand{\M}{\mathbb{M}}

\newcommand{\Pol}{\mathrm{Pol}}
\newcommand{\Trig}{\mathrm{Trig}}

\renewcommand{\P}{\mathbb{P}}
\newcommand{\Q}{\mathbb{Q}}

\newcommand{\Fc}{\mathcal{F}}
\newcommand{\A}{\mathcal{A}}
\newcommand{\TSASFunc}{\mathcal{L}}
\newcommand{\Hc}{\mathcal{H}}
\newcommand{\activF}{\sigma}

\renewcommand{\d}{\mathrm{d}}

\begin{document}

\title{Reservoir Computing Universality With Stochastic Inputs}

\author{Lukas~Gonon and
        Juan-Pablo~Ortega
\thanks{L. Gonon and J.-P. Ortega are with the Department
of Mathematics and Statistics, Universit\"at Sankt Gallen, Sankt Gallen, Switzerland. L. Gonon is also affiliated with the Department of Mathematics, ETH Z\"urich, Switzerland. J.-P. Ortega is also affiliated with the Centre National de la Recherche Scientifique (CNRS), France.}
}


\markboth{Gonon and Ortega: Reservoir Computing Universality With Stochastic Inputs}%
{Gonon and Ortega: Reservoir Computing Universality With Stochastic Inputs}
%


\maketitle

\begin{abstract}
The universal approximation properties with respect to $L ^p $-type criteria of three important families of reservoir computers with stochastic discrete-time semi-infinite inputs is shown. First, it is proved that linear reservoir systems with either polynomial or neural network readout maps are universal. More importantly, it is proved that the same property holds for two families with linear readouts, namely, trigonometric state-affine systems and echo state networks, which are the most widely used reservoir systems in applications. The linearity in the readouts is  a key feature  in supervised machine learning applications. It guarantees that these systems can be used in high-dimensional situations and in the presence of large datasets. The $L ^p $ criteria used in this paper allow the formulation of universality results that do not necessarily  impose almost sure uniform boundedness in the inputs or the fading memory property in the filter that needs to be approximated.
\end{abstract}


%
\IEEEpeerreviewmaketitle

\section{Introduction}
%
%
%
%
\IEEEPARstart{A}{universality} statement in relation to a machine learning paradigm refers to its versatility at the time of reproducing a rich number of patterns obtained by modifying only a limited number of hyperparameters. In the language of learning theory, universality amounts to the possibility of making approximation errors as small as one wants \cite{cucker:smale, Smale2003, cucker:zhou:book}. Well-known universality results are, for example, the uniform approximation properties of feedforward neural networks established in \cite{cybenko, hornik} for deterministic inputs and, later on, extended in \cite{Hornik1991} to accommodate random inputs.

This paper is a generalization of the universality statements in \cite{Hornik1991}  to a discrete-time dynamical context. More specifically, we are interested in the learning not of functions but  of filters that transform semi-infinite random input sequences parameterized by time into outputs that depend on those inputs in a  causal and time-invariant manner. The approximants used are small subfamilies of
reservoir computers (RC)~\cite{maass1, Jaeger04} or reservoir systems. Reservoir computers are filters generated by nonlinear state-space transformations and constitute special types of recurrent neural networks. They are determined by two maps, namely a {  reservoir} $F: \mathbb{R} ^N\times \mathbb{R} ^n\longrightarrow  \mathbb{R} ^N$, $n,N \in \mathbb{N} $,  and a readout map $h: \mathbb{R}^N \rightarrow \mathbb{R}$ that under certain hypotheses transform (or filter) an infinite discrete-time input  ${\bf z}=(\ldots, {\bf z} _{-1}, {\bf z} _0, {\bf z} _1, \ldots) \in (\mathbb{R}^n) ^{\Bbb Z } $ into an output signal ${\bf y} \in \mathbb{R}^{\Bbb Z } $ of the same type using  a state-space transformation given by:
\begin{empheq}[left={\empheqlbrace}]{align}
\mathbf{x} _t &=F(\mathbf{x}_{t-1}, {\bf z} _t),\label{eq:RCSystem1}\\
{y} _t &= h (\mathbf{x} _t), \label{eq:RCSystem2}
\end{empheq}
where $t \in \Bbb Z $ and the dimension $N \in \mathbb{N} $ of the state vectors $\mathbf{x} _t \in \mathbb{R} ^N $ is referred to as the number of virtual neurons of the system. In supervised machine learning applications the reservoir map is very often randomly generated and the memoryless readout is trained so that the output matches a given teaching signal.

Families of systems of this type have already been proved to be universal in different contexts. In the continuous-time setup, it was shown in \cite{Boyd1985} that linear reservoir systems with polynomial readouts or bilinear reservoirs with linear readouts are able to uniformly approximate any fading memory filter with uniformly bounded and equicontinuous inputs. The fading memory property is a continuity feature exhibited by many filters encountered in applications. 

In the discrete-time setup, several universality statements were already part of classical systems theory statements for inputs defined on a finite number of time points \cite{Sontag1979, sontag:polynomial:1979, FliessNormand1980}. In the more general context of semi-infinite inputs, various universality results have been formulated for systems with approximate finite memory \cite{sandberg:esn, sandberg:esn:paper, Matthews:thesis, Matthews1993, perryman:thesis, Stubberuda}. These universality results have been recently extended to the causal and fading memory category in \cite{rc6, rc7}. In those works it has been established the universality of two important families of reservoir systems with linear readouts, namely, the so called state affine systems (SAS) and the echo state networks (ESN). Moreover,  the universality of the SAS  family was established in \cite{rc6} both for uniformly bounded deterministic inputs, as well as for almost surely uniformly bounded stochastic ones. This last statement was shown to be a corollary of a general transfer theorem that proves that very important features of causal and time-invariant filters like the fading memory property or universality are naturally inherited by reservoir systems with almost surely uniformly bounded stochastic inputs from their counterparts with deterministic inputs. 

Unfortunately, almost surely bounded random inputs are not always appropriate for many applications. For example,  most parametric time series models use as driving innovations random variables whose distributions are not compactly supported (Gaussian, for example) in order to ensure adequate levels of performance. The main goal of this work is {\it formulating universality results in the stochastic context that do not impose almost sure uniform boundedness in the inputs}.

The way in which the universality results contained in this paper are articulated differs somewhat from the above quoted references and is more in the vein of  \cite{Hornik1991}. More specifically, in the stochastic universality statements in \cite{rc6}, for example, universal families are presented that uniformly approximate any given filter for any input in a given class of stochastic processes.  In contrast with this strategy and like in  \cite{Hornik1991}, we fix here first a discrete-time stochastic process that models the data generating process (DGP) behind the system inputs that are being considered. Subsequently, families of reservoir filters are spelled out whose images of the DGP are dense in the $L ^p $ sense. Equivalently, the image of the DGP by any measurable causal and time invariant filter can be approximated by the image of one of the members of the universal family with respect to an $L ^p $ norm defined using the law of the prefixed DGP. 

It is important to point out that this approach allows us to {\it formulate universality results for filters that do not necessarily have the fading memory property since only measurability is imposed as a hypothesis}.

The paper contains three main universality statements. The first one shows that linear reservoir systems with either polynomial or neural network readout maps are universal in the $L ^p $ sense. More importantly, two other families with linear readouts are shown to also have this property, namely, trigonometric state-affine systems and echo state networks, which are the most widely used reservoir systems in applications. The linearity of the readout is a key feature of these systems since in supervised machine learning applications it reduces the training task to the solution of a linear regression problem, which can be implemented efficiently also in high-dimensional situations and in the presence of large datasets.

\section{Preliminaries}

In this section we introduce some notation and collect general facts about filters, reservoir systems, and stochastic input signals.

\subsection{Notation} We write $\N=\{0,1,\ldots\}$ and $\Z_-=\{\ldots,-1,0 \}$. The elements of the Euclidean spaces  $\mathbb{R}^n $ will be written as column vectors and will be denoted in bold. Given a vector $\mathbf{v} \in \mathbb{R}  ^n $, we denote its entries by $v_i$ or by $v^{(i)} $, with $i \in \left\{ 1, \dots, n
\right\} $.
$(\R^n)^{\Z}$ and $(\R^n)^{\Z_-}$ denote the sets of infinite $\R^n$-valued sequences of the type $(\ldots, {\bf z}_{-1}, {\bf z}_0,{\bf z}_1,\ldots )$ and $(\ldots, {\bf z}_{-1}, {\bf z}_0)$ with ${\bf z}_i \in \R^n$ for $i \in \Z$ and $i \in \Z_-$, respectively. The elements in these sequence spaces will also be written in bold, for example, ${\bf z}:=(\ldots, {\bf z}_{-1}, {\bf z}_0) \in (\R^n)^{\Z_-}$. We denote by $\mathbb{M}_{n ,  m }$ the space of real $n\times m$ matrices with $m, n \in \mathbb{N} $. When $n=m$, we use the symbol $\mathbb{M}  _n $ to refer to the space of square matrices of order 
$n$. Random variables and stochastic processes will be denoted using upper case characters that will be bold when they are vector valued.

\subsection{Filters and functionals}

A filter is a map $U \colon (\R^n)^{\Z} \to \R^{\Z}$. It is called causal, if for any ${\bf z}, \mathbf{w} \in (\R^n)^{\Z}$ which satisfy ${\bf z}_{\tau} = {\bf w}_{\tau}$ for all $\tau \leq t$ for a given $t \in \Z$, one has that $U({\bf z})_t = U({\bf w})_t$. Denote by $T_{-\tau} \colon (\R^n)^{\Z} \to (\R^n)^{\Z}$ the time delay operator defined by $T_{-\tau}({\bf z})_t := {\bf z}_{t+\tau}$, for any $\tau \in \Bbb Z $. A filter $U$ is called time-invariant, if $T_{-\tau} \circ U = U \circ T_{-\tau}$ for all $\tau \in \Z$. 

Causal and time-invariant filters can be equivalently described  using their naturally associated functionals. We refer to a map $H \colon (\R^n)^{\Z_-} \to \R$ as a functional. Given a causal and time-invariant filter $U$, one  defines the  functional $H_U$  associated to it by setting $H_U({\bf z}):= U({\bf z}^e)_0$. Here ${\bf z}^e$ is an arbitrary extension of ${\bf z} \in (\R^n)^{\Z_-}$ to $(\R^n)^{\Z}$. $H_U$ does not depend on the choice of this extension since $U$ is causal. Conversely, given a functional $H$ one may define a causal and time-invariant filter $U_H \colon (\R^n)^{\Z} \to \R^{\Z}$ by setting $U_H({\bf z})_t := H(\pi_{\Z_-} \circ T_{-t}({\bf z}))$, where $\pi_{\Z_-} \colon (\R^n)^{\Z} \to (\R^n)^{\Z_-}$ is the natural projection.  One may verify that any causal and time-invariant filter can be recovered from its associated functional and conversely. Equivalently, $U = U_{H_U}$ and $H = H_{U_H}$. We refer to \cite{Boyd1985} for further details. 

If $U$ is causal and time-invariant, then for any ${\bf z} \in (\R^n)^{\Z}$ the sequence $U({\bf z})$ restricted to $\Z_-$ only depends on $({\bf z}_t)_{t \in \Z_-}$. Thus we may also consider $U$ as a map $U \colon (\R^n)^{\Z_-} \to \R^{\Z_-}$, but when we do so this will always be clear from the context. 

\subsection{Reservoir computing systems} 

A specific class of filters can be obtained using the reservoir computing systems or reservoir computers (RC) introduced in \eqref{eq:RCSystem1}-\eqref{eq:RCSystem2} when they satisfy the following property: a reservoir system satisfies the echo state property (ESP) if for any ${\bf z} \in (\R^n)^{\Z}$ there exists a unique $\mathbf{x} \in (\R^N)^{\Z}$ such that \eqref{eq:RCSystem1} holds. In this case the RC system gives rise to a filter $U^F_h $ associating to any ${\bf z} \in (\R^n)^{\Z}$ the unique output in \eqref{eq:RCSystem2}, that is, \ $U^F_h({\bf z})_t:= y_t$. Furthermore, it can be shown (see \cite[Proposition~2.1]{rc7}) that $U^F_h$ is necessarily  causal and time-invariant and hence we may associate to $U^F_h$ a reservoir functional $H^{F}_h \colon (\R^n)^{\Z_-} \to \R$ defined as $H^F_h({\bf z}) := U^F_h({\bf z})_0$. 

As seen above, the causal and time-invariant filter $U^F_h$ is uniquely determined by the reservoir functional $H^F_h$. Since the latter is determined by the restriction of the RC system to $\Z_-$,  
we will sometimes consider the system \eqref{eq:RCSystem1}-\eqref{eq:RCSystem2} only for $t \in \Z_-$.

\subsection{Deterministic filters with stochastic inputs}

We are interested in feeding the filters and the systems that we just introduced with stochastic processes as inputs. More explicitly, given a causal and time-invariant filter $U$ that satisfies certain measurability hypotheses, any stochastic process ${\bf Z} = ({\bf Z}_t)_{ t \in \Z_- }$ is mapped to a new stochastic process $(U({\bf Z})_t)_{t \in \Z_-}$. The main contributions in this article address the question of approximating $U({\bf Z})$ by reservoir filters in an $L^p$-sense. We now introduce the precise framework to achieve this goal.

\subsubsection{Probabilistic framework}

Consider a probability space $(\Omega,\Fc,\P)$ on which all random variables are defined. The input signal is modeled as a discrete-time stochastic process ${\bf Z} = ({\bf Z}_t)_{t \in \Z_-}$ with values in $\R^n$. When dealing with stochastic processes we will make no distinctions  between the assignment ${\bf Z}: \Bbb Z_- \times \Omega \rightarrow \mathbb{R}^n  $ and the corresponding  map into path space ${\bf Z}: \Omega \rightarrow ({\Bbb R}^n)^{\mathbb{Z}_-} $. We recall that ${\bf Z}$  is a stochastic process  when the corresponding map ${\bf Z}: \Omega \rightarrow ({\Bbb R}^n)^{\mathbb{Z}_-} $ is measurable. Here $(\R^n)^{\Z_-}$ is equipped with the product $\sigma$-algebra $\otimes_{t \in \Z_-} \mathcal{B}(\R^n)$ (which coincides with the Borel $\sigma$-algebra of $(\R^n)^{\Z_-}$ equipped with the product topology by \cite[Lemma 1.2]{Kallenberg2002}), where $\mathcal{B}(\R^n)$ is the Borel $\sigma$-algebra on ${\Bbb R}^n $.

We denote by  $\Fc_t := \sigma({\bf Z}_0,\ldots,{\bf Z}_{t})$, $t \in \Z_-$, the $\sigma$-algebra generated by $\left\{{\bf Z}_0,\ldots,{\bf Z}_{t}\right\} $ and write $\Fc_{-\infty} :=\sigma({\bf Z}_t \colon t \in \Z_-)$. For $p \in [1,\infty]$ we denote by $L^p(\Omega,\Fc,\P)$ the  Banach space formed by the real-valued random variables in $(\Omega,\Fc,\P)$ that have a finite usual $L^p$ norm $\| \cdot \|_p$.

We say that the process ${\bf Z} $ is stationary when for any $\{ t _1, \ldots, t _k\} \subset \mathbb{Z}_-   $, $h \in \Bbb Z_-  $, and $A _{t _1} , \ldots, A _{t _k} \in \mathcal{B}(\R^n) $, we have that \vspace{-0.5mm}
\begin{multline*}
\mathbb{P} \left({\bf Z}_{t _1} \in A _{t _1} , \ldots, {\bf Z}_{t _k} \in A _{t _k}\right)\\=\mathbb{P} \left({\bf Z}_{t _1+h} \in A _{t _1} , \ldots, {\bf Z}_{t _k+h} \in A _{t _k}\right).
\end{multline*}


\subsubsection{Measurable functionals and filters}

We say that a functional $H$ is measurable when the map between measurable spaces $H \colon \left((\R^n)^{\Z_-}, \otimes_{t \in \Z_-} \mathcal{B}(\R^n)\right) \to \left(\R, \mathcal{B}(\R) \right)$ is measurable. When $H$ is measurable then so is $H({\bf Z}):(\Omega, \mathcal{F}) \to (\R, \mathcal{B}(\mathbb{R}))$  since $H({\bf Z}) = H \circ {\bf Z}$ is the composition of measurable maps and hence  $H({\bf Z})$ is a random variable on $(\Omega,\Fc,\P)$. 

Analogously, we will say that a causal, time-invariant filter $U$ is measurable when the map between measurable spaces $U \colon \left((\R^n)^{\Z}, \otimes_{t \in \Z} \mathcal{B}(\R^n)\right) \to \left(\R^{\Z}, \otimes_{t \in \Z} \mathcal{B}(\R)\right)$ is measurable. In that case, also the restriction of $U$ to $\Z_-$ (see above) is measurable and so $U({\bf Z}) $ is a real-valued stochastic process.

As discussed above, causal, time-invariant filters and functionals are in a one-to-one correspondence. This relation is compatible with the measurability condition, that is, a causal and time-invariant filter is measurable if and only if the associated functional is measurable. In order to prove this statement we show first that 
the  operator $\pi _{\mathbb{Z}_{-}} \circ T _{-t}:  \left((\R^n)^{\Z}, \otimes_{t \in \Z} \mathcal{B}(\R^n)\right) \longrightarrow  \left((\R^n)^{\Z_-}, \otimes_{t \in \Z_-} \mathcal{B}(\R^n)\right)$ is a measurable map, for any $t \in \mathbb{Z}_{-}$.
Indeed, notice first that the projections $p _i: \left((\R^n)^{\Z}, \otimes_{t \in \Z} \mathcal{B}(\R^n)\right) \longrightarrow\left(\R^n,  \mathcal{B}(\R^n)\right) $, $i \in \mathbb{Z}_{-}  $, given by $p _i({\bf z})= {\bf z}_i $  are measurable. Since $\pi _{\mathbb{Z}_{-}}  \circ T _{-t} $ can be written as the  Cartesian product  of measurable maps $\pi _{\mathbb{Z}_{-}}  \circ T _{-t}=\prod _{i=- \infty}^{t} p _i = \left(\ldots,p_{t-2},p_{t-1}, p _t\right)$, it is hence measurable  \cite[Lemma 1.8]{Kallenberg2002}.

Now, if $H$ is a measurable functional, this implies that the associated filter \vspace{-1mm}
\begin{equation}
\label{eq:filterRepres}
U _H= \prod _{t=-\infty}^{0} H \circ  \pi _{\mathbb{Z}_{-}}  \circ T _{-t} 
\vspace{-0.5mm}
\end{equation}
is also measurable since  it is a composition of measurable functions. Conversely, if $U$ is causal, time-invariant, and measurable, then so is the associated functional $H _U=p _0 \circ U $.

\subsubsection{$L^p$-norm for functionals}
Fix $p \in [1,\infty)$ and let $H$ be a measurable functional such that   $H({\bf Z}) \in L^p(\Omega,\Fc,\P)$. The functionals which satisfy that
\begin{equation} \label{eq:distance} \|H({\bf Z})\|_p := \E[|H({\bf Z})|^p]^{1/p}< \infty \end{equation}
will be referred to as $p$-integrable with respect to the input process ${\bf Z}$.

Let us now consider the expression \eqref{eq:distance} from an alternative point of view. 
Denote by $\mu_{{\bf Z}}:= \P \circ {\bf Z}^{-1}$ the law of ${\bf Z}$ when viewed as a $(\R^n)^{\Z_-}$-valued random variable as above. Thus $\mu_{{\bf Z}}$ is a probability measure on $(\R^n)^{\Z_-}$ such that for any measurable set $A \subset (\R^n)^{\Z_-}$  one has $\mu_{{\bf Z}}(A)=\P({\bf Z} \in A)$. The requirement $H({\bf Z}) \in L^p(\Omega,\Fc,\P)$ then translates to $H \in L^p((\R^n)^{\Z_-},\mu_{{\bf Z}})$ and \eqref{eq:distance} is equal \cite[Lemma 1.22]{Kallenberg2002} to 
\begin{equation*}
\|H\|_p^{\mu_{{\bf Z}}}:= \left[ \int_{(\R^n)^{\Z_-}} |H({\bf z})|^p \mu_{{\bf Z}}(\d {\bf z}) \right]^{1/p} =\|H({\bf Z})\|_p.
\end{equation*}

Thus, the  results formulated later on in the paper for functionals with random inputs can also be seen as statements for functionals with deterministic inputs in $(\R^n)^{\Z_-} $,  where the closeness between them is  measured using the norm in $L^p((\R^n)^{\Z_-},\mu_{{\bf Z}})$. Following the terminology used by \cite{Hornik1991} we will refer to $\mu_{{\bf Z}}$ as the input environment measure.

We emphasize that these two points of view are equivalent. Given any probability measure $\mu_{{\bf Z}}$ on $(\R^n)^{\Z_-}$ one may set $\Omega = (\R^n)^{\Z_-}$, $\Fc= \otimes_{t \in \Z_-} \mathcal{B}(\R^n)$, $\P = \mu_{{\bf Z}}$ and define $Z_t({\bf z}):={\bf z}_t$ for all ${\bf z} \in \Omega$. We will switch between these two viewpoints throughout the paper without much warning to the reader.  


\subsubsection{$L^p$-norm for filters}

Fix $p \in [1,\infty)$. A causal, time-invariant, measurable filter $U$ is said to be $p$-integrable, if
\begin{equation}
\label{eq:filterNorm} 
\| U({\bf Z}) \|_p:= \sup_{t \in \Z_-}\left\{\E \left[  |U({\bf Z})_t|^p \right]^{1/p}\right\} < \infty. 
\end{equation}
It is easy to see that if $U$ is $p$-integrable, then so is the corresponding functional $H _U $ due to the following inequality
\begin{multline*}
\|H_U({\bf Z})\|_p=\E[|H_U({\bf Z})|^p]^{1/p}=\E[|U({\bf Z})_0|^p]^{1/p}\\
\leq \sup_{t \in \Z_-}\left\{\E \left[  |U({\bf Z})_t|^p \right]^{1/p}\right\}=\| U({\bf Z}) \|_p < \infty.
\end{multline*}

The converse implication holds true when the input process is stationary. In order to show this fact, notice first that if $\mu _t $ is the law of $\pi _{\mathbb{Z}_{-}}  \circ T _{-t}({\bf Z})$, $t \in \mathbb{Z}_{-}$,  and ${\bf Z} $ is by hypothesis stationary then,
for any $\{ t _1, \ldots, t _k\} \subset \mathbb{Z}_-   $ and $A _{t _1} , \ldots, A _{t _k} \in \mathcal{B}(\R^n) $, we have that
\begin{multline*}
\mathbb{P} \left((\pi _{\mathbb{Z}_{-}}  \circ  T _{-t}({\bf Z}))_{t _1} \in A _{t _1} , \ldots, (\pi _{\mathbb{Z}_{-}}  \circ  T _{-t}({\bf Z}))_{t _k} \in A _{t _k}\right)\\=\mathbb{P} \left({\bf Z}_{t _1+t} \in A _{t _1} , \ldots, {\bf Z}_{t _k+t} \in A _{t _k}\right)\\
=\mathbb{P} \left({\bf Z}_{t _1} \in A _{t _1} , \ldots, {\bf Z}_{t _k} \in A _{t _k}\right),
\end{multline*}
which proves that 
\begin{equation}
\label{eq:mutequalsmuz}
\mu_{{\bf Z} }= \mu_t, \quad \mbox{for all} \quad t \in \mathbb{Z}_{-}.
\end{equation}
This identity, together with \eqref{eq:filterRepres}, implies that for any $p$-integrable functional $H$:
\begin{multline}
\label{eq:equal norms}
\| U_H({\bf Z}) \|_p=\sup_{t \in \Z_-}\left\{\E \left[  |U_H({\bf Z})_t|^p \right]^{1/p}\right\}\\
=\sup_{t \in \Z_-}\left\{\E \left[  |H( \pi _{\mathbb{Z}_{-}}  \circ T _{-t}({\bf Z}))|^p \right]^{1/p}\right\}\\=
\sup_{t \in \Z_-}\left\{\left[ \int_{(\R^n)^{\Z_-}} |H({\bf z})|^p \mu_t(\d {\bf z}) \right]^{1/p}\right\}\\=\sup_{t \in \Z_-}\left\{\left[ \int_{(\R^n)^{\Z_-}} |H({\bf z})|^p \mu_{{\bf Z}}(\d {\bf z}) \right]^{1/p}\right\}= \left\|H({\bf Z})\right\|_p< \infty,
\end{multline}
which proves the $p$-integrability of the associated filter $U _H $.

%

\section{$L^p$-universality results}
Fix $p \in [1,\infty)$, ${\bf Z}  $ an input process, and a functional $H$ such that $H(\mathbf{Z}) \in L^p(\Omega,\Fc,\P)$. The goal of this section is finding simple families of reservoir systems that are able to approximate $H(\mathbf{Z})$ as accurately as needed in the $L^p$-sense. The first part contains a result that shows that linear reservoir maps with polynomial readouts are able to carry this out. The situation is hence identical to the case for deterministic inputs or for almost surely uniformly bounded stochastic ones \cite{rc6}. The second part contains a family that is able to achieve universality using only linear readouts, which is of major importance for applications since in that case the training effort reduces to solving a linear regression. Finally, we prove the universality of echo state networks which is the most widely used family of reservoir systems with linear readouts.

\subsection{Linear reservoirs with nonlinear readouts}

Consider a reservoir system with linear reservoir map and a polynomial readout. More precisely, given $A \in \M_N$, ${\bf c} \in \M_{N,n}$, and $h\in \Pol_N$ a real-valued polynomial in $N$ variables, consider the system 
\begin{equation}\label{eq:RCPolReadoutDet}
\left\{
\begin{aligned}
\mathbf{x}_t & = A \mathbf{x}_{t-1} + {\bf c} {\bf z}_t, \quad t \in \Z_-, \\ 
y_t & = h(\mathbf{x}_t), \quad t \in \Z_-,
\end{aligned}
\right.
\end{equation}
for any ${\bf z} \in (\R^n)^{\Z_-}$.
If the matrix $A$ is chosen so that $\sigma_{{\rm max}}(A)<1 $, then this system has the echo state property and the corresponding reservoir filter $U^{A,{\bf c}}_h $  is causal and time-invariant \cite{rc6}. We denote by $H^{A,{\bf c}}_h$ the associated functional. We are interested in the approximation capabilities that can be achieved by using processes of the type $H^{A,{\bf c}}_h(\mathbf{Z})$, where ${\bf Z} $ is a fixed input process and $H^{A,{\bf c}}_h(\mathbf{Z}) = Y_0$, with $Y_0$ obviously determined by the stochastic reservoir system 
\begin{equation}\label{eq:RCPolReadout}
\left\{
\begin{aligned}
\mathbf{X}_t & = A \mathbf{X}_{t-1} + c \mathbf{Z}_t, \quad t \in \Z_-, \\ 
Y_t & = h(\mathbf{X}_t), \quad t \in \Z_-.
\end{aligned}
\right.
\end{equation}

\begin{proposition}
\label{prop:Nonlinear} 
Fix $p \in [1,\infty)$, let
${\bf Z} $ be a fixed $\mathbb{R}^n $-valued input process, and let $H$ be a functional such that $H(\mathbf{Z}) \in L^p(\Omega,\Fc,\P)$. Suppose that for any $K \in \N$ there exists $\alpha >0$ such that 
\begin{equation}\label{eq:mixedExpMoment} 
 \E\left[\exp\left(\alpha \sum_{k=0}^K \sum_{i=1}^n |Z^{(i)}_{-k}|\right)\right] < \infty.
 \end{equation}
Then, for any $\varepsilon > 0$ there exists $N \in \N$, $A \in \M_N$, ${\bf c} \in \M_{N,n}$, and $h \in \Pol_N$ such that  \eqref{eq:RCPolReadoutDet} has the echo state property, the corresponding filter is causal and time-invariant, the associated functional satisfies $H^{A,c}_h(\mathbf{Z}) \in L^p(\Omega,\Fc,\P)$, and
\begin{equation}\label{eq:approxPolReadout}
\| H(\mathbf{Z}) - H^{A,c}_h(\mathbf{Z}) \|_p < \varepsilon.
\end{equation}
If the input process ${\bf Z} $ is stationary then 
\begin{equation}\label{eq:approxPolReadoutfilter}
\| U_H(\mathbf{Z}) - U^{A,c}_h(\mathbf{Z}) \|_p < \varepsilon.
\end{equation}
\end{proposition}

\begin{proof} The proof consists of two steps: In the first one we use assumption \eqref{eq:mixedExpMoment} and classical results in the literature to establish that 
\begin{equation}\label{eq:PolDense} \Pol_{n(K+1)} \text{ is dense in  } L^p(\R^{n(K+1)},\mu_K), \text{ for all  } K \in \N, \end{equation} 
where $\mu_K$ is the law of $(Z_0^{(1)},Z_0^{(2)},\ldots,Z_{-K}^{(n-1)},Z_{-K}^{(n)})$ on $\R^{n(K+1)}$ under $\P$.
In the second step we then use \eqref{eq:PolDense} to construct a linear RC system of the type in \eqref{eq:RCPolReadoutDet} that yields the approximation statement \eqref{eq:approxPolReadout}.

\textit{Step 1:} Denote by $\mu_K$ the law of $(Z_0^{(1)},Z_0^{(2)},\ldots,Z_{-K}^{(n-1)},Z_{-K}^{(n)})$ on $\R^{N}$ under $\P$, where $N:=n(K+1)$.
By \eqref{eq:mixedExpMoment} there exists $\alpha > 0$ such that $\int_{\R^{N}} \exp(\alpha \|{\bf z}\|_1) \mu_K(\d {\bf z}) < \infty$, where here and in the rest of this proof $\|\cdot\|_1$ denotes the Euclidean $1$-norm. Denoting by $\mu_K^{j}$ the $j$-th marginal distribution of $\mu_K$,  this implies for $j=1,\ldots,N$ that
\[ \int_\R \exp(\alpha |{z}^{(j)}|) \mu_K^j(\d {z}^{(j)}) \leq \int_{\R^{N}} \exp(\alpha \|{\bf z}\|_1) \mu_K(\d {\bf z}) < \infty.\]
Consequently,  by \cite[Theorem~6]{Berg1981}, $\Pol_1$ is dense in $L^p(\R,\mu_K^j)$ for any $p \in [1,\infty)$, $j=1,\ldots,N$. By \cite[Proposition page 364]{Petersen1983} this implies that $\Pol_N$ is dense in $L^p(\R^N,\mu_K)$, where we note that $\mu_K$ indeed satisfies the moment assumption in \cite[Page 361]{Petersen1983}: since $x^{2 m} \leq \exp(\alpha x)$ for any $x \geq 0$, $m \in \N$, one has
\[\begin{aligned} \int_{\R^{N}} \|{\bf z}\|_2^{2 m} \mu_K(\d {\bf z}) & \leq  \int_{\R^{N}} \exp(\alpha \|{\bf z}\|_2) \mu_K(\d {\bf z}) \\ & \leq  \int_{\R^{N}} \exp(\alpha \|{\bf z}\|_1) \mu_K(\d {\bf z}) < \infty. 
\end{aligned} \]

\textit{Step 2:} Let $\varepsilon > 0$. By Lemma~\ref{lem:condExp} in the appendix there exists $K \in \N$ such that 
\begin{equation}
\label{eq:auxEq0}
\|H({\bf Z})-\E[H({\bf Z})|\Fc_{-K}]\|_p < \frac{\varepsilon}{2}
\end{equation}
where $\Fc_{-K} := \sigma({\bf Z}_0,\ldots,{\bf Z}_{-K})$.
In the following paragraphs we will establish the approximation statement \eqref{eq:approxPolReadout} for $\E[H(Z)|\Fc_K]$ instead of $H(Z)$. Combining this with \eqref{eq:auxEq0} will then yield \eqref{eq:approxPolReadout}. 

Let $N:=n(K+1)$. By definition, $\E[H({\bf Z})|\Fc_{-K}]$ is $\Fc_{-K}$-measurable and hence there exists \cite[Lemma~1.13]{Kallenberg2002} a measurable function $g_K\colon \R^N \to \R$ such that $\E[H({\bf Z})|\Fc_{-K}] = g_K  ({\bf Z}_0,\ldots,{\bf Z}_{-K})$. Furthermore, 
\begin{multline*} \int_{\R^N} |g_K({\bf z})|^p \mu_K(\d {\bf z})  \\
= \E[|\E[H({\bf Z})|\Fc_{-K}]|^p]  \leq \E[|H({\bf Z})|^p]  < \infty,
\end{multline*}
by standard properties of conditional expectations (see, for instance, \cite[Theorem~5.1.4]{Durrett2010}) and the assumption that $H({\bf Z}) \in L^p(\Omega,\Fc,\P)$. Thus, $g_K \in L^p(\R^N,\mu_K)$ and using the statement \eqref{eq:PolDense} established in Step 1, there exists $h \in \Pol_N$ such that 
\begin{multline}
\label{eq:auxEq1}   
\|\E[H({\bf Z})|\Fc_{-K}]-h({\bf Z}_0^{\top},\ldots,{\bf Z}_{-K}^{\top})\|_p  \\= \|g_K - h\|_{L^p(\R^N,\mu_K)}  < \frac{\varepsilon}{2}.  
\end{multline}
Define now a reservoir system of the type \eqref{eq:RCPolReadout} with inputs given by the random variables ${\bf Z} _t  $, $t \in \mathbb{Z}_{-} $ and reservoir matrices $A \in \M_N$ and $c \in \M_{N,n}$ with all entries equal to $0$ except $A_{i,i-n}=1$ for $i=n+1,\ldots,N$ and $c_{i,i} = 1$
for $i=1,\ldots,n$, that is
\begin{equation*}
A= \left(
\begin{array}{cc}
 \boldsymbol{0}_{n,nK}&\boldsymbol{0}_{n,n}\\
  \boldsymbol{I}_{nK}&\boldsymbol{0}_{n,n}
\end{array}
\right), \quad \mbox{and} \quad
c= \left(
\begin{array}{c}
\boldsymbol{I}_{n}\\
\boldsymbol{0}_{nK,n}\\
\end{array}
\right).
\end{equation*}
This system has the echo state property (all the eigenvalues of $A$ equal zero) and has a unique causal and time invariant solution associated to the reservoir states ${\bf X} _t:= \left({\bf Z} _t^{\top}, {\bf Z}_{t-1}^{\top},\ldots , {\bf Z}_{t-K}^{\top}\right)^\top$, $t \in \mathbb{Z}_{-} $. It is easy to verify that the corresponding reservoir functional is given by
\begin{equation}
\label{eq:auxEq2} 
H^{A,c}_h({\bf Z}) = h({\bf Z}_0^{\top},\ldots,{\bf Z}_{-K}^{\top}). 
\end{equation}
Now the triangle inequality and \eqref{eq:auxEq0}, \eqref{eq:auxEq1} and \eqref{eq:auxEq2} allow us  to conclude \eqref{eq:approxPolReadout}.

The statement in \eqref{eq:approxPolReadoutfilter} in the presence of the stationarity hypothesis for $\mathbf{Z} $ is a straightforward consequence of \eqref{eq:mutequalsmuz} and the equality \eqref{eq:equal norms}.
\end{proof}

\begin{remark} \label{condition iid}
A sufficient condition for \eqref{eq:mixedExpMoment} to hold is that the random variables
$\{\mathbf{Z}_t \colon t \in \Z_-\}$ are independent and that for each $t$, there exists a constant $\alpha >0$ such that $\E[\exp(\alpha \sum_{i=1}^n |Z_t^{(i)}|)] < \infty$.
\end{remark}

\begin{remark}
Assumption \eqref{eq:mixedExpMoment} can be replaced by alternative assumptions but it can not be removed. Even if $n=1$ and $\{{\bf Z}_t \colon t \in \Z_-\}$ are independent and identically distributed with distribution $\nu$, a  condition \textit{stronger} than the existence of moments of all orders for $\nu$ is required. As a counterexample, one may take for $\nu$ a lognormal distribution. Then $\nu$ has moments of all orders, but \eqref{eq:mixedExpMoment} is not satisfied. Let us now argue that the approximation result proved under assumption~\eqref{eq:mixedExpMoment} fails in this case. The following argument relies on results for the classical \textit{moment problem} (see, for example, the collection of references in \cite{Ernst2012}).

Indeed, by \cite{Heyde1963} $\nu$ is not determinate (there exist other probability measures with identical moments) and thus (see e.g.\ \cite[Theorem~4.3]{Freud1971}) 
$\Pol_1$ is not dense in $L^p(\R,\nu)$ for $p \geq 2$. In particular, there exists $g \in L^p(\R,\nu)$ and $\varepsilon > 0$ such that $\|g - \tilde{h} \|_p > \varepsilon$ for all $\tilde{h} \in \Pol_1$. Suppose that we are in the case $n=1$ and let $\{{Z}_t \colon t \in \Z_-\}$ be independent and identically distributed with distribution $\nu$ and $H({\bf z}):=g({z}_0)$ for ${\bf z} \in \R^{\Z_-}$. Then, for any choice of $N$, $A$, $c$ and $h$ one has $\E[H^{A,c}_h({\bf Z})|\Fc_0]=\tilde{h}({Z}_0)$, where $\tilde{h}(x):=\E[h(A\mathbf{X}_{-1}+c x)], x \in \R$, is a polynomial. Thus one may  use \cite[Theorem~5.1.4]{Durrett2010} and the fact that by construction $H({\bf Z})$ is $\Fc_0$-measurable to obtain 
\[\begin{aligned} \|H({\bf Z})- H^{A,c}_h({\bf Z})\|_p & \geq \|\E[H({\bf Z})|\Fc_0]- \E[H^{A,c}_h({\bf Z})|\Fc_0]\|_p \\ & = \|g-\tilde{h}\|_p > \varepsilon.\end{aligned} \]
\end{remark}

\begin{remark}
In previous reservoir computing universality results for both deterministic and stochastic inputs quoted in the introduction there was an important continuity hypothesis called the fading memory property that does not play a role here and that has been replaced by the integrability requirement $H \in L^p((\R^n)^{\Z_-},\mu_{{\bf Z}})$. In particular, the universality results that we just proved and those that come in the next section (see Theorem \ref{thm:TSAS}) yield approximations for filters which do not necessarily have the fading memory property. Whether or not the approximation results apply depends on the integrability condition with respect to the input environment measure $\mu_{{\bf Z}}$. Consider, for example, the functional associated to the peak-hold operator \cite{Boyd1985}. In the discrete-time setting, the associated functional is
\[ H({\bf z})= \sup_{t \leq 0} {z}_t, \quad \mbox{with} \quad {\bf z} \in \mathbb{R}^{\mathbb{Z}_-}.\]
We now show that the two possibilities $H \in L^p((\R^n)^{\Z_-},\mu_{{\bf Z}})$ and $H \notin L^p((\R^n)^{\Z_-},\mu_{{\bf Z}})$ are feasible, depending on the choice of $\mu_{{\bf Z}}$:
\begin{itemize}
\item Let ${\bf Z}=(Z_t)_{t \in \Z_-}$ be one dimensional independent and identically distributed (i.i.d) random variables with unbounded support and denote by $\mu_{\bf Z}$ the law of ${\bf Z}$ on $\R^{\Z_-}$. Denoting by $F$ the distribution function of $Z_1$ and using the i.i.d assumption one calculates, for any $a \in \R$,
\[\begin{aligned} \P(H({\bf Z})>a) & =1-\P(\cap_{t < 0} \{Z_t \leq a\})\\ & =1-\lim_{n \to \infty} F(a)^n = 1. \end{aligned} \]
Hence, we can conclude that $H({\bf Z})=\infty$, $\mu_{{\bf Z}}$-almost everywhere and therefore $H \notin L^p((\R^n)^{\Z_-},\mu_{{\bf Z}})$.
\item Consider now the same setup, but  assume this time that the random variables have bounded support, that is, for some $a_{\mathrm{max}}\in \mathbb{R}$ one has that $P(Z_t \leq a_{\mathrm{max}})=1$ and $P(Z_t> a_{\mathrm{max}})=0$. Then, the same argument shows that $H({\bf Z})=a_{\mathrm{max}}$, $\mu_{{\bf Z}}$-almost everywhere and therefore $H \in L^p((\R^n)^{\Z_-},\mu_{{\bf Z}})$. 
\end{itemize}
\end{remark}

\begin{remark}
From the proof of Proposition~\ref{prop:Nonlinear} one sees that one could replace in its statement $\Pol_N$ by any other family $\{\Hc_N\}_{N \in \N}$ that satisfies the density statement \eqref{eq:PolDense}. In particular, the following corollary shows that this result can be obtained with readouts made out of neural networks. 
\end{remark}

Denote by  $\mathcal{H}_N$ the set feedforward one hidden layer neural networks with inputs in $\R^N$ that are constructed with a fixed activation function $\activF$. More specifically, $\mathcal{H}_N$ is made of functions $h \colon \R^N \to \R$ of the type
\begin{equation} 
\label{eq:NN}  
h({\bf x}) = \sum_{j=1}^k \beta_j \activF(\bm{\alpha}_j \cdot {\bf x} - \theta_j), 
\end{equation}
for some $k \in \N$, $\beta_j,\theta_j \in \R$, and $\bm{\alpha}_j \in \R^N$, for $j=1,\ldots,k$.

\begin{corollary}  \label{cor:nn}
In the setup of Proposition~\ref{prop:Nonlinear}, consider the family of neural networks $h \in \Hc_N$ constructed with a fixed activation function $\activF$ that is bounded and non-constant. Then, for any $\varepsilon > 0$ there exists $N \in \N$, $A \in \M_N$, $c \in \M_{N,n}$, and a neural network $h \in \Hc_N$ such that the corresponding reservoir system  \eqref{eq:RCPolReadoutDet} has the echo state property and has a unique causal and time-invariant filter associated. Moreover, the functional $H^{A,c}_h({\bf Z}) \in L^p(\Omega,\Fc,\P)$ and satisfies that
\begin{equation}\label{eq:approxNetworkReadout}
\| H({\bf Z}) - H^{A,c}_h({\bf Z}) \|_p < \varepsilon.
\end{equation}

\end{corollary}

\begin{proof}
By \cite[Theorem~1]{Hornik1991} the set $\Hc_N$ is dense in $L^p(\R^N,\mu)$ for any finite measure $\mu$ on $\R^N$. Thus, statement \eqref{eq:PolDense} holds with  $\Hc_N$ replacing $\Pol_{n(K+1)}$. Mimicking line by line the proof of Step~2 in Proposition~\ref{prop:Nonlinear} then proves the Corollary.
\end{proof}

\subsection{Trigonometric state-affine systems with linear readouts} Fix $M,N \in \N$ and consider $R \colon \R^n \to \M_{N,M}$ defined by
\begin{equation}
\label{eq:trigPolynomial} 
R({\bf z}):= \sum_{k=1}^r A_k \cos(\mathbf{u}_k \cdot {\bf z}) + B_k \sin(\mathbf{v}_k \cdot {\bf z}), \quad {\bf z} \in \R^n ,
\end{equation}
for some $r \in \N$, $A_k, B_k \in \M_{N,M}$, $\mathbf{u} _k, \mathbf{v} _k \in \R^n$, for $k=1,\ldots,r$. The symbol $\Trig_{N,M}$ denotes the set of all functions of the type \eqref{eq:trigPolynomial}. We  call the elements of $\Trig_{N,M}$ trigonometric polynomials.

We now introduce reservoir systems with linear readouts and reservoir maps constructed using trigonometric polynomials: let $N \in \N$, ${\bf W} \in \R^N$, $P \in \Trig_{N,N}$, $Q \in \Trig_{N,1}$ and define, for any ${\bf z} \in (\R^n)^{\Z_-}$, the system: 
\begin{equation}\label{eq:RCTSASDet}
\left\{
\begin{aligned}
\mathbf{x}_t & = P({\bf z}_t) \mathbf{x}_{t-1} + Q({\bf z}_t), \quad t \in \Z_-, \\ 
y_t & = {\bf W}^\top \mathbf{x}_t, \quad t \in \Z_-.
\end{aligned}
\right.
\end{equation}
We call the systems of this type trigonometric state-affine systems.
When such a system has the echo state property and a unique causal and time-invariant solution for any input, we denote by $U^{P,Q}_{\bf W} $ the corresponding filter and by  $H^{P,Q}_{\bf W}({\bf z}):= y_0$ the associated functional. As in the previous section,  we fix $p \in [1,\infty)$, ${\bf Z}  $ an input process, and a functional $H$ such that $H(\mathbf{Z}) \in L^p(\Omega,\Fc,\P)$ and we are interested in approximating $H(\mathbf{Z})$ by systems of the form $H^{P,Q}_{\bf W}(\mathbf{Z})$. Again, we will write $H^{P,Q}_{\bf W}(\mathbf{Z})=Y_0$, where $Y_0$ is uniquely determined by the reservoir system with stochastic inputs
\begin{equation}\label{eq:RCTSAS}
\left\{
\begin{aligned}
\mathbf{X}_t & = P(Z_t) \mathbf{X}_{t-1} + Q(Z_t), \quad t \in \Z_-, \\ 
Y_t & = {\bf W}^\top \mathbf{X}_t, \quad t \in \Z_-.
\end{aligned}
\right.
\end{equation}
Define $\A$ as the set of four-tuples $(N,{\bf W},P,Q)\in \N \times \R^N \times \Trig_{N,N} \times \Trig_{N,1}$ whose associated systems \eqref{eq:RCTSASDet} have the echo state property and the unique solutions are causal and time-invariant. In particular, for such $(N,{\bf W},P,Q)$ a reservoir functional $H^{P,Q}_{\bf W} $ associated to \eqref{eq:RCTSASDet} exists.

\begin{theorem}
\label{thm:TSAS}
Let $p \in [1,\infty)$ and let
${\bf Z} $ be a fixed $\mathbb{R}^n $-valued input process.
Denote by $\TSASFunc_{{\bf Z}}$ the set of reservoir functionals of the type \eqref{eq:RCTSASDet} which are $p$-integrable, that is,\
\[ \TSASFunc_{{\bf Z}} := \{ H^{P,Q}_{\bf W}({\bf Z}) \, : \, (N,{\bf W},P,Q) \in \A \} \cap L^p(\Omega,\Fc,\P).  \]
Then $\TSASFunc_{{\bf Z}}$ is dense in $L^p(\Omega,\Fc_{-\infty},\P)$. 

In particular, for any functional $H$ such that $H(\mathbf{Z}) \in L^p(\Omega,\Fc,\P)$ and any $\varepsilon > 0$, there exists $N \in \N$, ${\bf W} \in \R^N$, $P \in \Trig_{N,N}$ and $Q \in \Trig_{N,1}$ such that the system \eqref{eq:RCTSASDet} has the echo state property and causal and time-invariant solutions. Moreover, $H^{P,Q}_{\bf W}({\bf Z}) \in L^p(\Omega,\Fc,\P)$ and
\begin{equation}\label{eq:approxTSAS}
\| H({\bf Z}) - H^{P,Q}_{\bf W}({\bf Z}) \|_p < \varepsilon.
\end{equation}
If the input process ${\bf Z} $ is stationary then 
\begin{equation}\label{eq:approxPolReadoutfiltersas}
\| U_H(\mathbf{Z}) - U^{P,Q}_{\bf W}(\mathbf{Z}) \|_p < \varepsilon.
\end{equation}
\end{theorem}

\begin{proof}
We first argue that $\TSASFunc_{{\bf Z}}$ is a linear subspace of $L^p(\Omega,\Fc_{-\infty},\P)$. To do this we need to introduce some notation. Given $A \in \M_{N_1,M_1}$, $B \in M_{N_2,M_2}$, we denote by $A \oplus B \in \M_{N_1+N_2,M_1+M_2}$ the direct sum. Given $R$ as in  \eqref{eq:trigPolynomial} we define $R \oplus A \in \Trig_{N+N_1,M+M_1}$ by
\[ 
R \oplus A ({\bf z}):= \sum_{k=1}^r A_k \oplus A \cos(\mathbf{u}_k \cdot {\bf z}) + B_k \oplus A \sin(\mathbf{v}_k \cdot {\bf z}),
\]
and (with the analogous definition for $B \oplus R$) for $R_i \in \Trig_{N_i,M_i}$, $i=1,2$ we set 
\[ R_1 \oplus R_2 = R_1 \oplus \mathbf{0}_{N_2,M_2} + \mathbf{0}_{N_1,M_1} \oplus R_2. \]
One easily verifies that for $\lambda \in \R$ and $(N_i,{\bf W}_i,P_i,Q_i) \in \A$, $i=1,2$, one has that
\[ \begin{aligned} & (N_1+N_2,{\bf W}_1 \oplus \lambda {\bf W}_2, P_1 \oplus P_2, Q_1 \oplus Q_2 ) \in \A, \\
&  H^{P_1,Q_1}_{{\bf W}_1}({\bf Z}) + \lambda H^{P_2,Q_2}_{{\bf W}_2}({\bf Z}) = H^{P_1 \oplus P_2,Q_1 \oplus Q_2}_{{\bf W}_1 \oplus \lambda {\bf W}_2}({\bf Z}).
\end{aligned} \]
This shows that $\TSASFunc_{{\bf Z}}$ is indeed a linear subspace of $L^p(\Omega,\Fc_{-\infty},\P)$.

Secondly, in order to show that $\TSASFunc_{{\bf Z}}$ is dense in $L^p(\Omega,\Fc_{-\infty},\P)$, it suffices to prove that if $F \in L^q(\Omega,\Fc_{-\infty},\P)$ satisfies $\E[F H] = 0$ for all $H \in \TSASFunc_{{\bf Z}}$, then $F=0$, $\P$-almost surely. Here $q\in (1,\infty]$ is the H\"older conjugate exponent of $p$.
This can be shown by contraposition.
Suppose that $\TSASFunc_{{\bf Z}}$ is not dense in $L^p(\Omega,\Fc_{-\infty},\P)$. Since $\TSASFunc_{{\bf Z}}$ is a linear subspace, by the Hahn-Banach theorem there exists a bounded linear functional $\Lambda$ on $L^p(\Omega,\Fc_{-\infty},\P)$ such that $\Lambda(H)=0$ for all $H \in \TSASFunc_{{\bf Z}}$, but $\Lambda \neq 0$, see e.g.\ \cite[Theorem~5.19]{Rudin:real:analysis}. Then by \cite[Theorem~6.16]{Rudin:real:analysis} there exists $F \in L^q(\Omega,\Fc_{-\infty},\P)$ such that $\Lambda(H)=\E[F H]$ for all $H \in L^p(\Omega,\Fc_{-\infty},\P)$ and $F \neq 0$, since $\Lambda \neq 0$. In particular, there exists $F \in L^q(\Omega,\Fc_{-\infty},\P) \setminus \{0\}$ such that $\E[F H]=0$ for all $H \in \TSASFunc_{{\bf Z}}$. 

Thirdly, suppose that $F \in L^q(\Omega,\Fc_{-\infty},\P)$ satisfies 
\begin{equation} 
\label{eq:auxEq4} 
\E[F H] = 0 \text{ for all } H \in \TSASFunc_{{\bf Z}}. 
\end{equation} 
If we show that $F=0$, $\P$-almost surely, then the statement in the theorem follows by the argument in the second step. In order to prove that $F=0$, $\P$-almost surely, we first show that  \eqref{eq:auxEq4} implies the following statement: for any $K \in \N$, any subset $I \subset \mathcal{I}_K:= \{0,\ldots,K\}$, and any $\mathbf{u}_0,\ldots, \mathbf{u}_K \in \R^n$ it holds that 
\begin{equation} \label{eq:auxEq3} \E \left[ F \prod_{j \in I} \sin(\mathbf{u}_j \cdot {\bf Z}_j) \prod_{k \in \mathcal{I}_K \setminus I} \cos(\mathbf{u}_k \cdot {\bf Z}_k)  \right] = 0. \end{equation}
We prove this claim by induction on $K \in \N$. For $K=0$, one sets $Q_1({\bf z}):= \cos(\mathbf{u}_0 \cdot {\bf z})$ and $Q_2({\bf z}):= \sin(\mathbf{u}_0 \cdot {\bf z})$ and notices that
$(1,1,0,Q_i) \in \A $. Moreover, since the sine and cosine function are bounded, it is easy see that $Q_i({\bf Z}_0)= H^{0,Q_i}_1({\bf Z}_0) \in \TSASFunc_{{\bf Z}}$, for $i \in \{1,2\}$. Thus \eqref{eq:auxEq4} implies \eqref{eq:auxEq3} and so the statement holds for $K=0$. For the induction step, let $K \in \N \setminus \{0\}$ and assume the implication holds for $K-1$. We now fix $I$ and  $\mathbf{u}_0,\ldots, \mathbf{u}_K \in \R^n$ as above and prove \eqref{eq:auxEq3}. To simplify the notation we define for $k\in \{0,\ldots,K\}$ and ${\bf z} \in \R^n$ the function $g_k$ by
\[ g_k({\bf z}):= \begin{cases}  \sin(\mathbf{u}_k \cdot {\bf z}), \quad \text{ if } k \in I, \\  \cos(\mathbf{u}_k \cdot {\bf z}) , \quad \text{ if } k \in \mathcal{I}_K \setminus I. \end{cases} \]
To prove \eqref{eq:auxEq3}, we set $N:=K+1$, for $j\in \{1,\ldots,K\}$ define $A_j \in \M_N$ with all entries equal to $0$ except $(A_j)_{j+1,j}=1$, that is, $(A_{j})_{k,l}= \delta_{k,j+1} \delta_{l,j}$, $k,l \in \left\{1, \ldots, N\right\}$. Define now for ${\bf z} \in \R^n$
\begin{equation}
\label{eq:intermediate system}
\left\{
\begin{aligned} P({\bf z})&:=   \sum_{j=0}^{K-1} A_{K-j} g_j({\bf z}), \\
  Q({\bf z})&:=   e_1 g_K({\bf z}), \\
  {\bf W}&:=  e_{K+1}, 
  \end{aligned} 
 \right.
\end{equation}
where $e_j$ is the $j$-th unit vector in $\R^N$, that is, the only non-zero entry of $e_j$ is a $1$ in the $j$-th coordinate. By Lemma~\ref{lem:matrixLem} in the appendix, one has $A_{j_L} \cdots A_{j_0} =0$ for any $j_0,\ldots, j_L \in \{ 1,\ldots, K\}$ and $L\geq K$, since $j_L=j_0+L$ can not be satisfied. In other words, any product of more than $K$ factors of matrices $A^{(j)}$ is equal to $0$ and thus for any $L \in \N$ with $L \geq K$ and any ${\bf z}_0,\ldots,{\bf z}_L \in \R^n$ one has $P({\bf z}_0)\ldots P({\bf z}_L) = 0$. Using this fact and iterating \eqref{eq:RCTSASDet}, one obtains that the trigonometric state-affine system defined by the elements in \eqref{eq:intermediate system} has a unique solution given by
\begin{equation}
\label{eq:integration sas nilpotent}
\mathbf{x}_t = Q({\bf z}_t) + \sum_{j=1}^K P({\bf z}_t) \cdots P({\bf z}_{t-j+1}) Q({\bf z}_{t-j}).  
\end{equation}
In particular $(N,{\bf W},P,Q) \in \A$ and 
\begin{multline}
\label{eq:expression hwpq for induction}
H^{P,Q}_{{\bf W}}({\bf Z}) 
= \mathbf{X} _0\\
= {\bf W}^{\top}\left(Q({\bf Z}_0) + \sum_{j=1}^K P({\bf Z}_0) \cdots P({\bf Z}_{-j+1}) Q({\bf Z}_{-j})\right).
\end{multline}
The finiteness of the sum in \eqref{eq:expression hwpq for induction} and the boundedness of the trigonometric polynomials implies that
$H^{P,Q}_{{\bf W}}({\bf Z}) \in  \TSASFunc_{{\bf Z}} $.

We conclude the proof of the induction step with the following chain of equalities that
uses \eqref{eq:auxEq4} in the first one, the representation \eqref{eq:expression hwpq for induction} in the second one, and the choice of the vector ${\bf W} $ and the induction hypothesis in the last step:
\begin{equation} \label{eq:auxEq6} \begin{aligned} 0 & = \E[F H^{P,Q}_{\bf W}({\bf Z})] \\ 
& = \E[F {\bf W}^\top Q({\bf Z}_0)] \\ & \quad \quad +  \E[F {\bf W}^\top \sum_{j=1}^K P({\bf Z}_0) \cdots P({\bf Z}_{-j+1}) Q({\bf Z}_{-j})  ]  \\ 
& = \E[F {\bf W}^\top  P({\bf Z}_0) \cdots P({\bf Z}_{-K+1}) Q({\bf Z}_{-K})  ]. \end{aligned} \end{equation}
However, again by Lemma~\ref{lem:matrixLem} in the appendix, the only non-zero product of matrices $A_{j_{K-1}} \cdots A_{j_0}$ for $j_0, \ldots j_{K-1} \in \{1,\ldots, K\}$ takes place when $j_k = k+1$ for $k \in \{0,\ldots, K-1\}$. Therefore: 
\[ \begin{aligned} & P({\bf Z}_0)  \cdots P({\bf Z}_{-K+1}) \\ & = A_{K} g_0({\bf Z}_0) A_{K-1} g_1({\bf Z}_{-1}) \cdots A_{1} g_{K-1}({\bf Z}_{-K+1}). \end{aligned} \]
Combining this with \eqref{eq:auxEq6} and using the identity \eqref{eq:auxEq5} in Lemma~\ref{lem:matrixLem} in the appendix one obtains 
\[\begin{aligned} 0 & = \E[F e_{K+1}^\top A_{K} \cdots A_{1} e_1  \prod_{k=0}^K g_k({\bf Z}_{-k}) ] \\ & = \E[F \prod_{k=0}^K g_k({\bf Z}_{-k}) ], \end{aligned} \] 
which is the same as \eqref{eq:auxEq3}.

Fourthly, by standard trigonometric identities, the identity
\eqref{eq:auxEq3} established in the third step implies that for any $K \in \N$, 
\begin{equation}\label{eq:auxEq7} \E\left[ F \exp\left(i \sum_{j=0}^K \mathbf{u}_j \cdot {\bf Z}_j\right) \right]  = 0 \text{ for all } \mathbf{u}_0,\ldots, \mathbf{u}_K \in \R^n. \end{equation}
We claim that \eqref{eq:auxEq7} implies $F=0$, $\P$-almost surely and hence the statement in the theorem follows. This fact is a consequence of the \textit{uniqueness} theorem for characteristic functions (which is ultimately a consequence of the Stone-Weierstrass approximation theorem). See for instance \cite[Theorem~4.3]{Kallenberg2002} and the text below that result. To prove $F=0$, $\P$-almost surely, we denote by $F^+$ and $F^-$ the positive and negative parts of $F$. Then by \eqref{eq:auxEq7} one has $\E[F] = 0$, necessarily. Thus, if it does not hold that $F = 0$, $\P$-almost surely, then $c:=\E[F^+]=\E[F^-] > 0$ and one may define probability measures $\Q^+$ and $\Q^-$ on $(\Omega,\Fc)$ by setting $\Q^+(A):=c^{-1} \E[F^+ \mathbbm{1}_A]$ and $\Q^-(A):=c^{-1} \E[F^- \mathbbm{1}_A]$ for $A \in \Fc$. 
Denote by $\mu_K^+$ and $\mu_K^-$ the law in $\R^{n(K+1)} $ of the random variable
\[ \boldsymbol{{\cal Z}}_K:= ( {\bf Z}_0^{\top},{\bf Z}_{-1}^{\top},\ldots,{\bf Z}_{-K}^{\top})^{\top}\]
under $\Q^+$ and $\Q^-$. Then, the statement \eqref{eq:auxEq7} implies that for all $u \in \R^{n(K+1)}$,
\[ \int_{\R^{n(K+1)}} \exp(i u \cdot z) \mu_K^+(\d z) =  \int_{\R^{n(K+1)}} \exp(i u \cdot {\bf z}) \mu_K^-(\d {\bf z}). \]
By the uniqueness theorem for characteristic functions (see e.g.\ \cite[Theorem~4.3]{Kallenberg2002} and the text below) this implies that $\mu_K^+ = \mu_K^-$. Translating this statement back to random variables, this means that for any bounded and measurable function $g \colon \R^{n(K+1)} \to \R$ one has
\[ \begin{aligned} 0  = c \E_{\Q^+}[g(\boldsymbol{{\cal Z}}_K)] - c\E_{\Q^-}[g(\boldsymbol{{\cal Z}}_K)]
  = \E[F g(\boldsymbol{{\cal Z}}_K)],
 \end{aligned} \]
 which, by definition, means that $\E[F|\Fc_{-K}]=0$, $\P$-almost surely.
 Since $K \in \N$ was arbitrary and $F \in L^1(\Omega,\Fc_{-\infty},\P)$, one may combine this with $\lim_{t \to -\infty} \E[F|\Fc_t] = F$, $\P$-almost surely\ (see Lemma~\ref{lem:condExp}) to conclude $F=0$, as desired. 
 
The statement in \eqref{eq:approxPolReadoutfiltersas} in the presence of the stationarity hypothesis for $\mathbf{Z} $ is a straightforward consequence of \eqref{eq:mutequalsmuz} and the equality \eqref{eq:equal norms}.
\end{proof}

We emphasize that the use in the proof of the theorem of nilpotent matrices of the type introduced in Lemma~\ref{lem:matrixLem} ensures that the the echo state property is automatically satisfied (see \eqref{eq:integration sas nilpotent}). 

%

\subsection{Echo state networks}

We now turn to showing the universality in the $L ^p$ sense of the  the most widely used reservoir  systems with linear readouts, namely,  echo state networks. An echo state network is a RC system determined by
\begin{equation} \label{eq:RCESNDet}
\left\{
\begin{aligned}
 \mathbf{x}_t &  = \activF( A \mathbf{x}_{t-1} + C {\bf z}_t + \bm{\zeta}), \\
y_t & = {\bf W}^\top \mathbf{x}_t,
\end{aligned}
\right.
\end{equation}
for $A \in \M_{N}$, $C \in M_{N,n}$, $\bm{\zeta} \in \R^N$, and ${\bf W} \in \R^N$. As it is customary in the neural networks literature, the map $\activF: \mathbb{R}^N\to \mathbb{R}^N$ is obtained via the componentwise application of a given activation function $\activF \colon \R \to \R$ that is denoted with the same symbol.

If this system has the echo state property and the resulting filter is causal and time-invariant, we write as $H^{A,C,\bm{\zeta}}_{\bf W}({\bf z}):= y_0$ the associated functional.

\begin{theorem}\label{thm:ESN} Fix $p \in [1,\infty)$, let
${\bf Z} $ be a fixed $\mathbb{R}^n $-valued input process, and let $H$ be a functional such that $H(\mathbf{Z}) \in L^p(\Omega,\Fc,\P)$. Suppose that the activation function $\activF \colon \R \to \R$ is  non-constant, continuous, and has a bounded image. Then for any $\varepsilon > 0$, there exists $N \in \N$, $C \in \M_{N,n}$, $\bm{\zeta} \in \R^N$, $A \in \M_{N}$, ${\bf W} \in \R^{N}$ such that \eqref{eq:RCESNDet} has the echo state property, the corresponding filter is causal and time-invariant, the associated functional satisfies $H^{A,C,\bm{\zeta}}_{\bf W}({\bf Z}) \in L^p(\Omega,\Fc,\P)$ and
\begin{equation}\label{eq:approxESN}
\| H({\bf Z}) - H^{A,C,\bm{\zeta}}_{\bf W}({\bf Z}) \|_p < \varepsilon.
\end{equation}
\end{theorem}

\begin{proof}
First, by Corollary~\ref{cor:nn} and \eqref{eq:auxEq2} there exists $K, \overline{N} \in \N$, $\overline{{\bf W}} \in \R^{\overline{N}}$, $\overline{A} \in \M_{\overline{N},n(K+1)}$, and $\overline{\bm{\zeta}} \in \R^{\overline{N}}$ such that the neural network  
\[ h({\bf z}) = \overline{{\bf W}}^\top \activF(\overline{A} {\bf z} + \overline{\bm{\zeta}}) \]
satisfies
\begin{equation} \label{eq:auxEq10} \| H({\bf Z}) - h({\bf Z}_0^\top,\ldots,{\bf Z}_{-K}^\top) \|_p < \frac{\varepsilon}{2}. \end{equation}
Notice that we may rewrite $\overline{A}$ as
\[ \overline{A} = [A^{(0)} A^{(-1)} \cdots A^{(-K)}] \]
with $A^{(j)} \in \M_{\overline{N},n}$ and 
\begin{equation}
\label{def hinfinity}
\begin{aligned} H_\infty({\bf Z}): & = h({\bf Z}_0^\top,\ldots,{\bf Z}_{-K}^\top) \\ & = \overline{{\bf W}}^\top \activF \left( \sum_{j=0}^K A^{(-j)} {\bf Z}_{-j} + \overline{\bm{\zeta}} \right). \end{aligned}
\end{equation}

Second, by the neural network approximation theorem for continuous functions \cite[Theorem~2]{Hornik1991}, for any $m \in \N$ there exists a neural network that uniformly approximates the identity mapping on the hypercube $B_m:=\{ { \bf x} \in \R^{n} \, : \, |x_i| \leq m \text{ for } i=1,\ldots,n \}$. More specifically, \cite[Theorem~2]{Hornik1991} is formulated for $\R$-valued mappings and we hence apply it componentwise: for any $m \in \N$ and $i=1,\ldots,n$ there exists $N_i^{(m)} \in \N$, ${\bf W}_i^{(m)} \in \R^{N_i^{(m)}}$, $\overline{A}_i^{(m)} \in \M_{N_i^{(m)},n}$, and $\bm{\zeta}_i^{(m)} \in \R^{N_i^{(m)}}$, such that for all $i=1,\ldots,n$ the neural network 
\[ h_i^{(m)}({ \bf x})= \left({\bf W}_i^{(m)}\right)^\top \activF\left(\overline{A}_i^{(m)} { \bf x} + \bm{\zeta}_i^{(m)}\right) \]
satisfies 
\begin{equation} \label{eq:auxEq8} \sup_{{\bf x} \in B_m} |h_i^{(m)}({\bf x}) -x_i| < \frac{1}{m}. \end{equation}
Write $h^{(m)}({\bf x})=(h_1^{(m)}({\bf x}),\ldots,h_n^{(m)}({\bf x}))^\top$ and for $j=1,\ldots, K$, denote by $[h^{(m)}]^j= h^{(m)} \circ \cdots \circ h^{(m)} $ the $j$th composition of $h^{(m)}$. We now claim that for all $j=1,\ldots, K$ and ${\bf x} \in \R^n$ it holds that 
\begin{equation}\label{eq:hmconv} \lim_{m \to \infty} [h^{(m)}]^j({\bf x}) = {\bf x}. \end{equation}
Indeed, let us fix ${\bf x} \in \R^n$ and argue by induction on $j$. To prove \eqref{eq:hmconv} for $j=1$, let $\overline{\varepsilon} > 0$ be given and choose $m_0 \in \N$ satisfying $m_0>\max\left\{|x_1|,\ldots,|x_n|,1/\overline{\varepsilon}\right\}$. Then, for any $m \geq m_0$ one has ${\bf x} \in B_m$ by definition and \eqref{eq:auxEq8} implies that for $i=1,\ldots,n$,
\[ |h_i^{(m)}({\bf x}) -x_i| < \frac{1}{m} < \overline{\varepsilon}. \]
Hence \eqref{eq:hmconv} indeed holds for $j=1$. Now let $j \geq 2$ and assume that \eqref{eq:hmconv} has been proved for $j-1$. Define $\overline{{\bf x}}^{(m)}:= [h^{(m)}]^{j-1}({\bf x})$. Then,  by the induction hypothesis,  for any given $\overline{\varepsilon} > 0$ one finds $m_0 \in \N$ such that for all $m \geq m_0$ and $i=1,\ldots, n$ it holds that 
\begin{equation} \label{eq:auxEq9} | \overline{ x}^{(m)}_i - x_i | < \frac{\overline{\varepsilon}}{2}. \end{equation}
Hence, choosing $\overline{m}_0 \in \N$ with
$\overline{m}_0>\max(m_0,|x_1|+\frac{\overline{\varepsilon}}{2},\ldots,|x_n|+\frac{\overline{\varepsilon}}{2},2/\overline{\varepsilon})$ one obtains from the triangle inequality and \eqref{eq:auxEq9} that $\overline{{\bf x}}^{(m)} \in B_{\overline{m}_0}$ for all $m \geq \overline{m}_0$. In particular for any $m \geq \overline{m}_0$ one may use the triangle inequality in the first step, $\overline{{\bf x}}^{(m)} \in B_{\overline{m}_0} \subset B_m$ and \eqref{eq:auxEq9} in the second step and \eqref{eq:auxEq8} in the last step to estimate 
\[\begin{aligned} |[h^{(m)}]^j_i({\bf x})-x_i| & \leq |h_i^{(m)}(\overline{{\bf x}}^{(m)})-\overline{x}^{(m)}_i| + |\overline{x}^{(m)}_i - x_i| \\ & \leq \sup_{{\bf y} \in B_m} |h_i^{(m)}({\bf y})-y_i| + \frac{\overline{\varepsilon}}{2} \\ & < \frac{1}{m} + \frac{\overline{\varepsilon}}{2} < \overline{\varepsilon}. \end{aligned} \]
This proves \eqref{eq:hmconv} for all $j=1,\ldots,K$.

Thirdly, define
\[ H_m({\bf Z}):=  \overline{{\bf W}}^\top \activF \left( \sum_{j=0}^K A^{(-j)} [h^{(m)}]^j({\bf Z}_{-j}) + \overline{\bm{\zeta}} \right) \]
with the convention $[h^{(m)}]^0({\bf x}) = {\bf x}$. 

Since $\activF$ is continuous, \eqref{eq:hmconv} implies that $\lim_{m \to \infty} H_m({\bf Z}) = H_\infty({\bf Z})$, $\P$-almost surely, where $H_\infty$ was defined in \eqref{def hinfinity}. Furthermore, by assumption there exists $C > 0$ such that $|\activF(x)| \leq C$ for all $x \in \R$. Hence one has $|H_\infty({\bf Z})-H_m({\bf Z})|^p \leq (2 C \sum_{i=1}^{\overline{N}}
|\overline{W}_i|)^p$ for all $m \in \N$. Thus one may apply the dominated convergence theorem to obtain
\[ \begin{aligned} & \lim_{m \to \infty} \|H_\infty({\bf Z})-H_m({\bf Z})\|_p \\ =  & \lim_{m \to \infty} \E[|H_\infty({\bf Z})-H_m({\bf Z})|^p]^{1/p} = 0 . \end{aligned} \]
In particular for $m \in \N$ large enough one has $\|H_\infty({\bf Z})-H_m({\bf Z})\|_p < \frac{\varepsilon}{2}$ and combining this with the triangle inequality and \eqref{eq:auxEq10} one obtains 
\begin{equation}\label{eq:auxEq11} \begin{aligned} \|H({\bf Z})- H_m({\bf Z})\|_p & \leq \|H({\bf Z})-H_\infty({\bf Z})\|_p \\ & \quad +\|H_\infty({\bf Z})-H_m({\bf Z})\|_p < \varepsilon. \end{aligned} \end{equation}

To conclude the proof we now fix $m \in \N$ large enough (so that \eqref{eq:auxEq11} holds) and show that $H_m({\bf Z})= H^{A,C,\bm{\zeta}}_{\bf W}({\bf Z})$
for suitable choices of $A,C,\bm{\zeta}$ and ${\bf W}$. To do so, first define $N_J:=N_1^{(m)}+\cdots+N_n^{(m)}$ and the block matrices
\[ W_J:=  \begin{pmatrix}
  ({\bf W}_1^{(m)})^\top &\multicolumn{2}{c}{\text{\kern0.5em\smash{\raisebox{-1ex}{\LARGE 0}}}}\\
    &  \ddots &   \\
\multicolumn{2}{c}{\text{\kern-2.0em\smash{\raisebox{-.5ex}{\LARGE 0}}}}& ({\bf W}_n^{(m)})^\top
 \end{pmatrix} \in \M_{n,N_J}, \]
\[  \bm{\zeta}_J:=  \begin{pmatrix}
  \bm{\zeta}_1^{(m)} \\
   \vdots    \\
   \bm{\zeta}_n^{(m)}
 \end{pmatrix} \in \R^{N_J}, \text{ and } A_J :=  \begin{pmatrix}
  \overline{A}_1^{(m)} \\
   \vdots    \\
   \overline{A}_n^{(m)}
 \end{pmatrix} \in \M_{N_J,n}. \]
Furthermore, to emphasize that $m$ is fixed and $h^{(m)}$ approximates the identity, set $J({\bf x}):=h^{(m)}({\bf x})$ and note that
\begin{equation}\label{eq:IDef} J({\bf x})= W_J \activF(A_J {\bf x} + \bm{\zeta}_J) .\end{equation}
Now set $N:=K N_J + \overline{N}$ and define the block matrix $A \in \M_{N}$ by
\[\! \! \! \! \! \! \! \! \! \! \! \! \! \! \! \! \! \! \! \! \! \! 
A = \begin{pmatrix}
  \boldsymbol{0}_{N_J,N_J} & & & \\
   A_J W_J & \boldsymbol{0}_{N_J,N_J} & & \\
    & A_J W_J & \ddots &   \multicolumn{1}{c}{\text{\kern 4em\smash{\raisebox{3ex}{\Huge 0}}}}\\
  \multicolumn{1}{c}{\text{\kern1.5em\smash{\raisebox{1ex}{\Huge 0}}}} &  & \ddots & \boldsymbol{0}_{N_J,N_J} \\
   &  &   & A_J W_J & \boldsymbol{0}_{N_J,N_J} \\
   A^{(-1)} W_J & A^{(-2)} W_J & \cdots & \cdots & A^{(-K)} W_J & \boldsymbol{0}_{\overline{N},\overline{N}}
 \end{pmatrix} \] 
and $\bm{\zeta} \in \R^{N}$, $C \in \M_{N,n}$ and ${\bf W} \in \R^N$ by
\[  \bm{\zeta}:=  \begin{pmatrix}
  \bm{\zeta}_J \\
   \vdots    \\
   \bm{\zeta}_J \\
   \overline{\bm{\zeta}}
 \end{pmatrix}, \quad  C:=  \begin{pmatrix}
  A_J \\
  \boldsymbol{0} \\ 
   \vdots    \\
  \boldsymbol{0} \\   
   A^{(0)}
 \end{pmatrix}, \text{ and } {\bf W}:=  \begin{pmatrix}
  \boldsymbol{0}_{K N_J,1} \\ \overline{{\bf W}} 
 \end{pmatrix}.\] 
 Furthermore, we partition the reservoir states $\mathbf{x}_t$ of the corresponding echo state system as
\[ \mathbf{x}_t:=  \begin{pmatrix}
  \overline{\mathbf{x}}_t^{(1)} \\ 
   \vdots    \\
  \overline{\mathbf{x}}_t^{(K+1)}
 \end{pmatrix}, \] 
 with $\overline{\mathbf{x}}_t^{(j)} \in \R^{N_J}$, for $j \leq K$, and $\overline{\mathbf{x}}_t^{(K+1)} \in \R^{\overline{N}}$. With this notation for $\mathbf{x}_t$ and these choices of matrices, the recursions associated to the echo state reservoir map in \eqref{eq:RCESNDet} read as
 \begin{align} \label{eq:ESNNilpotFirst} \overline{\mathbf{x}}_t^{(1)} & =  \activF(A_J {\bf z}_t + \bm{\zeta}_J), \\ 
 \label{eq:ESNNilpot} \overline{\mathbf{x}}_t^{(j)} & = \activF(A_J W_J \overline{\mathbf{x}}_{t-1}^{(j-1)} + \bm{\zeta}_J), \text{ for } j=2,\ldots,K, \\
 \label{eq:ESNNilpotLast}
\overline{\mathbf{x}}_t^{(K+1)} & = \activF(\sum_{j=1}^K A^{(-j)} W_J \overline{\mathbf{x}}_{t-1}^{(j)} + A^{(0)} {\bf z}_t + \overline{\bm{\zeta}}).
\end{align} 
By iteratively inserting \eqref{eq:ESNNilpot} into itself and using \eqref{eq:ESNNilpotFirst} one obtains (recall the definition of $J$ in \eqref{eq:IDef}) that the unique solution to \eqref{eq:ESNNilpot} is given by
\begin{equation}\label{eq:ESNRecursionSolved} \overline{\mathbf{x}}_t^{(j)} = \activF(A_J [J]^{j-1}({\bf z}_{t-j+1}) + \bm{\zeta}_J). \end{equation}
More formally, one uses induction on $j$: For $j=1$ the two expressions \eqref{eq:ESNRecursionSolved} and \eqref{eq:ESNNilpotFirst} coincide. For $j=2,\ldots,K$ one inserts \eqref{eq:ESNRecursionSolved} for $j-1$ (which holds by induction hypothesis) into \eqref{eq:ESNNilpot} to obtain
\[\begin{aligned} \overline{\mathbf{x}}_t^{(j)}  & = \activF(A_J W_J \activF(A_J [J]^{j-2}({\bf z}_{t-j+1}) + \bm{\zeta}_J) + \bm{\zeta}_J) \\ & =  \activF(A_J [J]^{j-1}({\bf z}_{t-j+1}) + \bm{\zeta}_J), \end{aligned} \]
which is indeed \eqref{eq:ESNRecursionSolved}. Finally, combining \eqref{eq:ESNRecursionSolved} and \eqref{eq:ESNNilpotLast} one obtains
\[\begin{aligned} y_t = \overline{{\bf W}}^\top & \overline{\mathbf{x}}_t^{(K+1)}  = \overline{{\bf W}}^\top \activF(\sum_{j=1}^K A^{(-j)} W_J \overline{\mathbf{x}}_{t-1}^{(j)} + A^{(0)} {\bf z}_t + \overline{\bm{\zeta}})
\\ & = \overline{{\bf W}}^\top \activF(\sum_{j=1}^K A^{(-j)} [J]^{j}({\bf z}_{t-j}) + A^{(0)} {\bf z}_t + \overline{\bm{\zeta}}).
 \end{aligned} \]
The statement \eqref{eq:ESNRecursionSolved} shows, in particular, that the echo state network associated to  $A,C,\bm{\zeta}$ and ${\bf W}$ satisfies the echo state property. Moreover,
inserting $t=0$ in the previous equality and comparing with the definition of $H_m({\bf Z})$ one sees that indeed $H_m({\bf Z})= H^{A,C,\bm{\zeta}}_{\bf W}({\bf Z})$. The approximation statement \eqref{eq:approxESN} therefore follows from \eqref{eq:auxEq11}.
\end{proof}

\subsection{An alternative viewpoint} 
So far all the universality results have been formulated for functionals and filters with random inputs. Equivalently, we may formulate them as $L^p$-approximation results on the sequence space $(\R^n)^{\Z_-}$ endowed with any measure $\mu $ that makes $p$-integrable the filter that we want to approximate.

\begin{theorem}\label{thm:sum} Let $H \colon (\R^n)^{\Z_-} \to \R$ be a measurable functional. Then, for any probability measure $\mu$ on $(\R^n)^{\Z_-}$ with $H \in L^p((\R^n)^{\Z_-},\mu)$ and any $\varepsilon > 0$ there exists  a reservoir system that has the echo state property and such that the corresponding filter is causal and time-invariant, the associated functional $H^{\text{RC}} $ satisfies that $H^{\text{RC}} \in L^p((\R^n)^{\Z_-},\mu)$ and
\begin{equation}\label{eq:approxPathspace}
\| H - H^{\text{RC}} \|_{L^p((\R^n)^{\Z_-},\mu)} < \varepsilon.
\end{equation}
The reservoir functional $H^{\text{RC}} $ may be chosen as coming from any of the following systems:
\begin{itemize}
\item Linear reservoir with polynomial readout, that is, \eqref{eq:RCPolReadoutDet} for some $N \in \N$, $A \in \M_N$, ${\bf c} \in \M_{N,n}$, and a polynomial $h \in \Pol_N$, if the measure $\mu$ satisfies the following condition: for any $K \in \N$, 
 \[\int_{(\R^n)^{\Z_-}} \exp\left(\alpha \sum_{k=0}^K \sum_{i=1}^n |z^{(i)}_{-k}|\right) \mu(\d z) < \infty. \]
\item Linear reservoir with neural network readout, that is, \eqref{eq:RCPolReadoutDet} for some $N \in \N$, $A \in \M_N$, ${\bf c} \in \M_{N,n}$, and a neural network $h \in \Hc_N$.
\item Trigonometric state-affine system with linear readout, that is, \eqref{eq:RCTSASDet} for some $N \in \N$, ${\bf W} \in \R^N$, $P \in \Trig_{N,N}$ and $Q \in \Trig_{N,1}$.
\item Echo state network with linear readout, that is, \eqref{eq:RCESNDet} for some $N \in \N$, $C \in \M_{N,n}$, $\bm{\zeta} \in \R^N$, $A \in \M_{N}$, ${\bf W} \in \R^{N}$, where we assume that $\activF\colon \R \to \R$ employed in \eqref{eq:RCESNDet} is bounded, continuous and non-constant. 
\end{itemize}
\end{theorem}
\begin{proof}
Set $\Omega = (\R^n)^{\Z_-}$, $\Fc= \otimes_{t \in \Z_-} \mathcal{B}(\R^n)$, $\P = \mu$ and define ${\bf Z}_t({\bf z}):={\bf z}_t$ for all ${\bf z} \in \Omega$, $t \in \Z_-$. Then $\Fc = \sigma({\bf Z}_t \, : \, t \in \Z_-)=\Fc_{-\infty}$ and ${\bf Z}$ is the identity mapping on $(\R^n)^{\Z_-}$. One may now apply Proposition~\ref{prop:Nonlinear}, Corollary~\ref{cor:nn}, Theorem~\ref{thm:TSAS} and Theorem~\ref{thm:ESN} with this choice of probability space $(\Omega,\Fc,\P)$ and input process ${\bf Z}$. The statement of Theorem~\ref{thm:sum} then precisely coincides with the statement of Proposition~\ref{prop:Nonlinear}, Corollary~\ref{cor:nn}, Theorem~\ref{thm:TSAS} and Theorem~\ref{thm:ESN}, respectively.
\end{proof}

\subsection{Approximation of stationary strong time series models}

Most parametric time series models commonly used in financial, macroeconometrics, and forecasting applications are specified by relations of the type
\begin{equation}
\label{eq:general form time series}
\mathbf{X} _t= G \left(\mathbf{X} _{t-1}, {\bf Z} _t, \boldsymbol{\theta}\right),
\end{equation}
where $\boldsymbol{\theta} \in \mathbb{R} ^k $ are the parameters of the model and the vector $\mathbf{X} _t \in \mathbb{R}^N $ is built so that it contains in its components the time series of interest and that, at the same time, allows for a Markovian representation of the model as in \eqref{eq:general form time series}. The model is driven by the innovations process ${\bf Z} = ({\bf Z}_t)_{ t \in \Z }\in \left({\Bbb R}^n\right)^{\Bbb Z}$. When the innovations are made out of  independent and identically distributed random variables we say that the model is strong \cite{Francq2010}. It is customary in the time series literature to impose constraints on the parameters vector $\boldsymbol{\theta}  $ so that the relation \eqref{eq:general form time series} has a unique second-order stationary solution or, in the language of this paper, the system   \eqref{eq:general form time series} satisfies the echo state property and the associated filter $U _G: \left({\Bbb R}^n\right)^{\Bbb Z} \rightarrow \left({\Bbb R}^N\right)^{\Bbb Z} $ satisfies that
\begin{equation}
\label{eq:stationary conditions}
\E \left[  U_G({\bf Z})_t \right]=: \boldsymbol{\mu} \mbox{ and }
\E \left[ U_G({\bf Z})_t U_G({\bf Z})_t^{\top} \right]=:\mbox{\boldmath$\Sigma   $} ,\, t \in \mathbb{Z}_{-},
\end{equation}
with $\boldsymbol{\mu} \in \mathbb{R}^N $  and $\mbox{\boldmath$\Sigma   $} \in \mathbb{M}_N  $ constant. The Wold decomposition theorem\cite[Theorem 5.7.1]{BrocDavisYellowBook}  shows that any such filter can be uniquely written as the sum of a linear and a deterministic process.

It is obvious that for strong models the stationarity condition \eqref{eq:mutequalsmuz} holds and that, moreover, the condition \eqref{eq:stationary conditions} implies that
\begin{equation}
\label{eq:norm ug}
\left\|U _G({\bf Z})\right\| _2=\sup_{t \in \Z_-}\left\{\E \left[  |U({\bf Z})_t|^2 \right]^{1/2}\right\}=
{\rm trace} \left(\boldsymbol{\Sigma}\right)^{1/2}< \infty.
\end{equation}
This integrability condition guarantees that the approximation results in Proposition \ref{prop:Nonlinear}, Corollary \ref{cor:nn}, and Theorems \ref{thm:TSAS} and \ref{thm:ESN} hold for second-order stationary strong time series models with $p=2 $. More specifically, the processes determined by this kind of models can be approximated in the $L ^2  $ sense by linear processes with polynomial or neural network readouts (when the condition in Remark \ref{condition iid} is satisfied), by trigonometric state-affine systems with linear readouts, or by echo state networks.

Important families of models to which this approximation statement can be applied are, among many others, (see the references for the meaning of the acronyms) GARCH \cite{engle:arch, bollerslev:garch}, VEC~\cite{bollerslev:vec}, BEKK~\cite{engle:bekk}, CCC~\cite{bollerslev:ccc}, DCC~\cite{tse:dcc, engle:dcc}, GDC~\cite{kroner:ng}, and ARSV~\cite{Taylor1982, harvey:ruiz:shephard}.

\section{Conclusion}

We have shown the universality of three different families of reservoir computers with respect to the $L ^p $ norm associated to any given discrete-time semi-infinite input process.

On the one hand we proved that linear reservoir systems with either neural network or, if the input process satisfies the exponential moments condition \eqref{eq:mixedExpMoment}, polynomial readout maps are universal. 

On the other hand we showed that this hypothesis can be dropped by considering two different reservoir families with linear readouts, namely, trigonometric state-affine systems and echo state networks. The latter are the most widely used reservoir systems in applications. The linearity in the readouts is  a key feature  in supervised machine learning applications of these systems. It guarantees that they can be used in high-dimensional situations and in the presence of large datasets, since the training in that case is reduced to a linear regression. 

We emphasize that, unlike existing results in the literature \cite{rc6, rc7} dealing with uniform universal approximation, the $L ^p $ criteria used in this paper allow to formulate universality statements that do not necessarily  impose almost sure uniform boundedness on the inputs or the fading memory property on the filter that needs to be approximated.


%

\appendix
\subsection{Auxiliary Lemmas}
\begin{lemma} \label{lem:condExp}
Let ${\bf Z}: \Bbb Z \times \Omega \rightarrow \mathbb{R}^n  $ be a stochastic process and let $\Fc_t := \sigma({\bf Z}_0,\ldots,{\bf Z}_{t})$, $t \in \Z_-$,  and  $\Fc_{-\infty} :=\sigma({\bf Z}_t \colon t \in \Z_-)\} $.
Let $F \in L^p(\Omega,\Fc_{-\infty},\P)$. Then $\E[F|\Fc_t]$ converges to $F$  as $t \to -\infty$, both $\P$-almost surely and in norm $\|\cdot\|_p$, for any $p \in [1,\infty)$. 
\end{lemma}

\begin{proof} Since $\Fc_{-t} \subset \Fc_{-t-1} \subset \Fc_{-\infty}$, for all $t\in \N$, and $F \in L^p(\Omega,\Fc_{-\infty},\P) \subset L^1(\Omega,\Fc_{-\infty},\P)$, one has by L\'evy's Upward Theorem (see, for instance,\ \cite[II.50.3]{Rogers2000} or \cite[Theorem~5.5.7]{Durrett2010}) that $F_t:=\E[F|\Fc_t]$ converges for $t \to -\infty$ to $F$ in $\|\cdot\|_1$ and $\P$-almost surely.\ If $p=1$ this already implies the claim. For $p>1$ one has 
by standard properties of conditional expectations (see, for instance,\ \cite[Theorem~5.1.4]{Durrett2010}) that $\sup_{t \in \N} \E[|F_t|^p] \leq \E[|F|^p]$. Hence \cite[Theorem~5.4.5]{Durrett2010} implies that $F_t$ converges for $t \to -\infty$ to some $\tilde{F} \in L^p(\Omega,\Fc_{-\infty},\P)$ both in $\|\cdot\|_p$ and $\P$-almost surely. But this identifies $\tilde{F}=\lim_{t \to -\infty} F_t = F$, $\P$-almost surely and hence $F_t$ converges for $t \to -\infty$ to $F$ also in $\|\cdot\|_p$.
\end{proof}

\begin{lemma}\label{lem:matrixLem}
For $N \in \N \setminus \{0,1\}$ and $j=1,\ldots,N-1$ define $A_{j} \in \M_{N}$ by $(A_{j})_{k,l}= \delta_{k,j+1} \delta_{l,j}$ for  $k,l \in \{ 1,\ldots,N\}$. Then
for $L \in \N$, $j_0,\ldots, j_L \in \{ 1,\ldots, N-1\}$ it holds that 
\begin{equation} \label{eq:auxEq5} (A_{j_L} \cdots A_{j_0})_{k,l} = \delta_{k,j_L + 1} \delta_{l,j_0} \prod_{i=1}^L \delta_{j_i,j_{i-1}+1}. \end{equation}
In particular $A_{j_L} \cdots A_{j_0} \neq 0$ if and only if $j_i=j_0+i$ for $i\in \{1,\ldots, L\}$.
\end{lemma}

\begin{proof}
The last statement directly follows from \eqref{eq:auxEq5}. To prove \eqref{eq:auxEq5} we proceed by induction on $L$.
Indeed, for $L=0$ the formula \eqref{eq:auxEq5} is just the definition of $A_{j_0}$. For the induction step, one assumes that \eqref{eq:auxEq5} holds for $L-1$ and calculates
\[\begin{aligned} (A_{j_L} & \cdots A_{j_0})_{k,l} \\ & = \sum_{r=1}^{N} \delta_{k,j_L+1} \delta_{r,j_L} (A_{j_{L-1}} \cdots A_{j_0})_{r,l} \\ & = \sum_{r=1}^{N} \delta_{k,j_L+1} \delta_{r,j_L} \delta_{r,j_{L-1} + 1} \delta_{l,j_0} \prod_{i=1}^{L-1} \delta_{j_i,j_{i-1}+1},  \end{aligned}  \]
which is indeed \eqref{eq:auxEq5}.  
\end{proof}


\section*{Acknowledgment}

The authors thank Lyudmila Grigoryeva and Josef Teichmann for helpful discussions and remarks and acknowledge partial financial support  coming from the Research Commission of the Universit\"at Sankt Gallen, the Swiss National Science Foundation (grants number 175801/1 and 179114), and the French ANR ``BIPHOPROC" project (ANR-14-OHRI-0018-02).

\ifCLASSOPTIONcaptionsoff
  \newpage
\fi



\bibliographystyle{IEEEtran}
\bibliography{/Users/JPO/Dropbox/Public/Gonon_Ortega}


%







\end{document}